\documentclass[12pt]{article}

\usepackage{amsmath,amsfonts,amsthm,amssymb,graphicx,epic,eepic,multicol}

\theoremstyle{definition}
\newtheorem{definition}{Definition}
\theoremstyle{remark}
\newtheorem{remark}[definition]{Remark}
\newtheorem{example}[definition]{Example}

\newtheoremstyle{mytheorem}{0.5cm}{0.2cm}{\slshape}{ }{\bfseries}{.}{ }{}
\theoremstyle{mytheorem}
\newtheorem{theorem}[definition]{Theorem}
\newtheorem{proposition}[definition]{Proposition}

\newtheorem{corollary}[definition]{Corollary}

\newcommand{\E}{\mathbf{E}}

\newcommand{\R}{\mathbb{R}}

\newcommand{\HH}{\mathbb{H}}

\DeclareMathOperator{\one}{{1\hspace*{-0.55ex}I}}
\DeclareMathOperator{\onev}{{\mathbf{1}}}

\DeclareMathOperator{\WSI}{WS}

\newcommand{\Q}{\mathbf{Q}}

\newcommand{\sH}{\mathcal{H}}

\newcommand{\fF}{\mathfrak{F}}
\newcommand{\salg}{\fF}

\renewcommand{\phi}{\varphi}
\renewcommand{\kappa}{\varkappa}
\newcommand{\imagi}{\boldsymbol{\imath}}
\newcommand{\kappat}{\tilde{\kappa}}

\newcommand{\gammaQ}{\gamma}

\newcommand{\nuQ}{\nu}

\newcommand{\fb}{f_{\mathrm{b}}}
\newcommand{\fbo}{f_{\mathrm{b}}^{\mathrm{o}}}

\newcommand{\thf}[1][1]{\frac{#1}{2}}

\newlength{\querylen}
\usepackage{fancybox}

\numberwithin{equation}{section}
\numberwithin{definition}{section}

\begin{document}
\bibliographystyle{plain}

\title{
  Exchangeability type properties of asset prices
\footnote{Supported by the
    Swiss National Science Foundation Grants Nr.  200021-117606 and 200021-126503.}}
\date{}

\author{Ilya Molchanov and Michael Schmutz\\
  \small Department of Mathematical Statistics and Actuarial Science,\\
  \small University of Bern, Sidlerstrasse 5, 3012 Bern, Switzerland\\
  \small (e-mails: ilya.molchanov@stat.unibe.ch, michael.schmutz@stat.unibe.ch)
}

\maketitle

\vspace{-5mm}
\begin{abstract}
  Let $\eta=(\eta_1,\dots,\eta_n)$ be a positive random vector. If its
  coordinates $\eta_i$ and $\eta_j$ are exchangeable, i.e.\ the
  distribution of $\eta$ is invariant with respect to the swap
  $\pi_{ij}$ of its $i$th and $j$th coordinates, then $\E f(\eta)=\E
  f(\pi_{ij}\eta)$ for all integrable functions $f$. This paper
  studies integrable random vectors that satisfy this identity for a
  particular family of functions $f$, namely those which can be
  written as the positive part of the scalar product $\langle
  u,\eta\rangle$ with varying weights $u$. In finance such functions
  represent payoffs from exchange options with $\eta$ being the random
  part of price changes, while from the geometric point of view they
  determine the support function of the so-called zonoid of $\eta$.
  If the expected values of such payoffs are $\pi_{ij}$-invariant, we
  say that $\eta$ is $ij$-swap-invariant. A full characterisation of
  the swap-invariance property and its relationship to the symmetries
  of expected payoffs of basket options are obtained. The first of
  these results relies on a characterisation theorem for integrable
  positive random vectors with equal zonoids. A particular attention
  is devoted to the case of asset prices driven by L\'evy processes.
  Based on this, concrete semi-static hedging techniques for
  multi-asset barrier options, such as weighted barrier swap options,
  weighted barrier quanto-swap options or certain weighted barrier
  spread options are suggested.

  \medskip

  \noindent


\end{abstract}

\section{Introduction}
\label{sec:introduction}

The classic univariate European \emph{put-call symmetry} property,
also known as Bates' rule from~\cite{bat97}, relates certain calls and
puts in the \emph{same} market, see e.g.~\cite{bat91,car94,faj:mor06b}
and more recently~\cite{car:lee08,teh09}. This symmetry property of an
integrable random variable $\eta$ can be expressed using expected
payoffs from plain vanilla options as
\begin{equation}
  \label{eq:put-call}
  \E(F\eta-k)_+=\E(F-k\eta)_+
\end{equation}
for every strike $k\geq 0$, with $F$ being the forward price, so
that the terminal asset price in the one-period setting is $F\eta$
(in order that the discounted expectations can be interpreted as
arbitrage-free prices, they have to be taken with respect to a
martingale measure), see \cite{car:lee08,mol:sch10}. In cases with
vanishing carrying costs the put-call symmetry makes it possible
to replace at certain times a call option with equally valued puts
in order to design so-called semi-static hedges for barrier
options. In cases of non-vanishing carrying costs semi-static
hedges can be constructed on the basis of a very closely related
property, called quasi-self-duality, being briefly discussed in
Section~\ref{sec:quasi-swap}, see also~\cite{car:lee08,mol:sch10}.

Following Carr and Lee~\cite{car:lee08}, \emph{semi-static
hedging} is the replication of contracts by trading European-style
claims at no more than two times after inception. In the
\emph{single asset} case such semi-static hedging strategies have
been analysed extensively in recent years,
see~e.g.~\cite{and:and:eli02,and01,bow:car94,car:cho97,car:ell:gup98}
and more recently~\cite{car:lee08}.

Interestingly, also the \emph{duality principle} in option pricing
traces some of its roots to the same papers as put-call symmetry
results, see e.g.~\cite{bat91,bat97,car94,grab83}. The power of
duality lies in the possibility to reduce the complexity of
valuation problems by relating them to easier problems in the
\emph{dual markets}. For a presentation of this
principle in a general univariate exponential semimartingale
setting see~\cite{eber:pap:shir08}, for bivariate L\'evy markets
see~\cite{faj:mor06c}, for multivariate semi-martingale extensions
(with various dual-markets) see~\cite{eber:pap:shir08b}. The
symmetry property then appears if the original and certain dual
markets coincide, that motivates the name \emph{self-dual} chosen
in~\cite{mol:sch10} for distributions that coincide with their
duals.

In the multi-asset setting $\eta=(\eta_1,\dots,\eta_n)$ is an
$n$-dimensional random vector with positive coordinates such that the
price $S_{Ti}$ of the $i$th asset at time $T>0$ equals $F_i\eta_i$,
where in a risk-neutral world, $F_i$ stands for the corresponding
theoretical \emph{forward price} and $\eta_i$ for the random part of
the price change of the $i$th asset. We denote this shortly as
\begin{displaymath}
  S_T=(S_{T1},\dots,S_{Tn})=(F_1\eta_1,\dots,F_n\eta_n)=F\circ\eta\,.
\end{displaymath}
Furthermore, assume that $\Q$ is a probability measure that makes
$\eta$ integrable. For later applications to barrier options we
extra assume that $\Q$ is a martingale measure that is consistent
with market option prices. The expectation with respect to $\Q$ is
denoted by $\E$ without subscript. For further simplicity of
notation, we do not write time $T$ as a subscript of $\eta$ and
incorporate for a moment the forward prices $F_i$, $i=1,\dots,n$,
into payoff functions. In our context payoff functions are
measurable functions $f:(0,\infty)^n\mapsto\R_+$.

Molchanov and Schmutz~\cite{mol:sch10} studied symmetries of expected
payoffs from European \emph{basket options} defined as
\begin{equation}
  \label{eq:fb}
  \fb(u_0,u_1,\dots,u_n)
  =\Big(\sum_{l=1}^n u_l\eta_l+u_0\Big)_+,\quad
  u_0, u_1,\dots,u_n\in\R\,.
\end{equation}
When writing the ``weights'' of a basket option together with its
strike as a vector, we number the coordinates of the obtained
$(n+1)$-dimensional vectors as $0,1,\dots,n$ and denote these
vectors as $(u_0,u)$ for $u_0\in\R$ and $u\in\R^n$ or as
$(u_0,u_1,\dots,u_n)=(u_0,u)\in\R^{n+1}$. In the following we
consider vectors as rows or columns depending on the situation.

Since $\fb(u_0,u)$ can be understood as a plain vanilla option on the
scalar product $\langle u,\eta\rangle$ with strike $u_0$, the
corresponding expected payoffs determine uniquely the distribution of
$\langle u,\eta\rangle$, see e.g.~\cite{bre:lit78,ros76}, and
thereupon also determine the distribution of $\eta$ as the following
result (which holds also for not necessarily positive $\eta$) shows.
Note that the expected values of $\fb(u_0,u)$ considered a function of
$(u_0,u)$ constitute the support function of an $(n+1)$-dimensional
convex body called the \emph{lift zonoid} of $\eta$, see \cite{mos02}.

\begin{theorem}[see e.g.~\cite{mol:sch10,mos02}]
  \label{thr:basket}
  The expected values $\E \fb(u_0,u_1,\dots,u_n)$ for all $u_0\in\R$ and
  $u\in\R^n$ determine uniquely the distribution $\Q$ of an
  integrable random vector $\eta$.
\end{theorem}

Although it is possible to weaken the statement of
Theorem~\ref{thr:basket} by considering only one fixed $u_0\neq 0$,
the uniqueness does not hold any more if $u_0=0$, i.e.\ for the
payoffs from swap (or exchange) options defined as
\begin{equation}
  \label{eq:fbo}
  \fbo(u)=\Big(\sum_{l=1}^n u_l\eta_l\Big)_+=(\langle
  u,\eta\rangle)_+\,,
  \quad u\in\R^n\,.
\end{equation}

The random vector $\eta$ with positive coordinates is called
\emph{self-dual} with respect to the $i$th numeraire if $\eta$ is
integrable and $\E \fb(u_0,u_1,\dots,u_n)$ as a function of
$(u_0,u)$ is invariant with respect to the permutation of $u_0$
and the $i$th coordinate of $u$, see~\cite[Sec.~2]{mol:sch10}. A
\emph{jointly} self-dual $\eta$ satisfies this property for all
numeraires $i=1,\dots,n$, so that the expected payoff $\E
\fb(u_0,u_1,\dots,u_n)$ becomes symmetric in all its $(n+1)$
arguments. This joint self-duality property implies that $\eta$ is
\emph{exchangeable}, i.e. $(\eta_1,\dots,\eta_n)$ coincides in
distribution with $(\eta_{l_1},\dots,\eta_{l_n})$ for each
permutation of its components.  The exchangeability property is
well studied in probability theory, see e.g.~\cite{ald85}
or~\cite{kal05s} and the literature cited therein. It is
also known from \cite[Sec.~3]{mol:sch10} that the exchangeability
property is strictly weaker than the joint self-duality.

While the self-duality property is crucial to switch between put and
call options as in (\ref{eq:put-call}), hedges for some other
derivatives do not rely on the self-duality assumption. In particular,
this relates to derivatives with the payoff function (\ref{eq:fbo}).
For example, one can require that
\begin{equation}
  \label{eq:two-swap}
  \E (u_1\eta_1+u_2\eta_2)_+=\E (u_1\eta_2+u_2\eta_1)_+
\end{equation}
for every $(u_1,u_2)\in\R^2$ in the two-asset case. This
swap-invariance property is weaker than the exchangeability of $\eta$,
e.g.\ it will be shown later that in the risk-neutral setting each
two-dimensional log-normally distributed random vector satisfies
(\ref{eq:two-swap}), no matter that its coordinates are not
identically distributed and so are not exchangeable unless the two
assets share the same volatility. This property helps to design
semi-static hedges for certain barrier options, e.g.\ the knock-out
contract with payoff defined by
\begin{displaymath}
  (aS_{T1}-bS_{T2})_+\one_{cS_{t1}>S_{t2}\forall t\in[0,T]}\,,
\end{displaymath}
where $S_{t1}$, $S_{t2}$, $t\in[0,T]$, are two price processes
(with equal carrying costs), for details see
Section~\ref{sec:hedg}, in particular
Example~\ref{ex:barrier-swap}.

We proceed with a concise discussion of the $ij$-exchangeability
property in Section~\ref{sec:exchangeability}.
Section~\ref{sec:ex-self-duality} characterises the weaker
swap-invariance property and discusses its relationships to
self-duality. Weighted variants of the swap-invariance are
considered in Section~\ref{sec:weight-swap-invar}.
Section~\ref{sec:swap-invar-infin} analyses log-infinitely
divisible distributions, exhibiting the swap-invariance property.
The necessity to handle unequal carrying costs in important
applications motivates further weakening of the swap-invariance
property in Section~\ref{sec:quasi-swap}. Finally
Section~\ref{sec:hedg} presents applications for creating
semi-static hedges for certain multi-asset derivatives with
knocking conditions.
The development of semi-static replication strategies of multi-asset
barrier options (see Section~\ref{sec:hedg}) and possibly also more
complicated path-dependent contracts is the probably most important
application of exchangeability type properties in finance. The
importance of developing robust hedging strategies for multi-asset
path-dependent financial derivatives is particularly stressed by Carr
and Laurence~\cite{car:laur09}.  Other obvious applications of the
described symmetry results may be found in the area of validating
models or analysing market data, e.g.\ extending the univariate case
considered in~\cite{bat97} and~\cite{faj:mor06}.

\section{Exchangeable random vectors}
\label{sec:exchangeability}

For each $i,j\in\{1,\dots,n\}$, $i\neq j$, define a linear mapping
on $\R^n$ by
\begin{displaymath}
  \pi_{ij} (x)=(x_1,\dots,x_{i-1},x_j,x_{i+1},\dots,
  x_{j-1},x_i,x_{j+1},\dots, x_n)\,,
\end{displaymath}
i.e.\ $\pi_{ij}$ transposes (swaps) the $i$th and $j$th
coordinates of $x$.
If the distribution of a random vector $\eta$ in $\R^n$ is
$\pi_{ij}$-invariant, we say that $\eta$ is \emph{$ij$-exchangeable}.
The following result follows directly from Theorem~\ref{thr:basket}.

\begin{corollary}
  \label{cor:ij-exch}
  An \emph{integrable} random vector $\eta$ is $ij$-exchangeable if
  and only if $\E\fb(u_0,u)$ is invariant with respect to permutation
  of the $i$th and $j$th coordinates of $u$ for all $u\in\R^n$ and
  any fixed $u_0\neq0$.
\end{corollary}

In view of financial applications assume that all coordinates of
$\eta$ are positive, so that $\eta=e^\xi$ for a random vector
$\xi=(\xi_1,\dots,\xi_n)$, where the exponential function is applied
coordinatewisely.  Because of the widespread use of L\'evy models for
derivative pricing we characterise infinitely divisible random vectors
$\xi=\log\eta$ for $ij$-exchangeable $\eta$.  In the sequel we denote
the Euclidean norm by $\|\cdot\|$, the imaginary unit $\sqrt{-1}$ by
$\imagi$, and use the following formulation of the
\emph{L\'evy-Khintchine formula} for the characteristic function of
$\xi$, see~\cite[Ch.~2]{sat99},
\begin{multline}
  \label{eq:levy-k}
  \varphi_\xi(u)=\E e^{\imagi\langle u,\xi\rangle}=
  \exp\bigg\{\imagi\langle\gamma,u\rangle-\thf\langle u,Au\rangle\\
  +\int_{\R^n}(e^{\imagi\langle u,x\rangle}-1
  -\imagi\langle u,x\rangle\one_{\| x\|\leq 1})d\nu(x)\bigg\}\,,\quad
  u\in\R^n\,,
\end{multline}
where $A$ is a symmetric non-negative definite $n\times n$ matrix,
$\gamma\in\R^n$ is a constant vector and $\nu$ is a L\'evy measure on
$\R^n$, namely $\nu(\{0\})=0$ and
\begin{equation}
  \label{eq:nu-conditions}
  \quad\int_{\R^n}\min(\|x\|^2,1)d\nu(x)<\infty\,.
\end{equation}

Since the $ij$-exchangeability of $\xi$ is equivalent to the
$\pi_{ij}$-invariance of its characteristic function, we immediately
obtain the following result.

\begin{proposition}
  \label{th:inv-div-CE}
  Let $\eta=e^\xi$ with $\xi$ being infinitely divisible.  Then $\eta$
  is $ij$-exchangeable if an only if the generating triplet
  $(A,\nu,\gamma)$ of $\xi$ satisfies the following conditions.
  \begin{itemize}
  \item[(1)] The matrix $A=(a_{lm})_{lm=1}^n$ satisfies
    $a_{ii}=a_{jj}$ and $a_{li}=a_{lj}$ for all $l=1,\dots,n$, $l\neq
    i,j$.
  \item[(2)] The L\'evy measure is $\pi_{ij}$-invariant, i.e.\
    $\nu(B)=\nu(\pi_{ij}B)$ for all Borel $B$.
  \item[(3)] The $i$th and $j$th coordinates of  $\gamma$
    coincide.
  \end{itemize}
\end{proposition}

\begin{example}[Log-normal distribution, Black--Scholes setting]
 \label{ex:log-n-mult}
 Assume that $\eta=e^\xi$ is log-normal with $\xi$ having expectation
 $\mu$ and covariance matrix $A$.
 Then $\eta$ is $ij$-exchangeable if and only if $A$ satisfies
 $a_{ii}=a_{jj}$ and $a_{li}=a_{lj}$ for $l=1,\dots,n$, $l\neq i,j$,
 (with the remaining $a_{lm}$ arbitrarily chosen such that $A$ is
 non-negative-definite) and $\mu_i=\mu_j$. The latter automatically
 holds if all components of $\eta$ are related to a martingale
 measure, i.e.\ $\mu=-\thf(a_{11},\dots,a_{nn})$. In
 \emph{bivariate risk-neutral} cases the only restriction is the
 equality of the variances, while the correlation coefficient between
 $\xi_1$ and $\xi_2$ can be arbitrary.
\end{example}

\section{Swap-invariance}
\label{sec:ex-self-duality}

Now we consider the symmetry property for the payoff function
(\ref{eq:fbo}).

\begin{definition}
  \label{def:si}
  An integrable random vector $\eta$ with positive components is
  said to be \emph{$ij$-swap-invariant} if the expected value $\E\fbo(u)$ is
  invariant with respect to swapping the $i$th and $j$th coordinates
  of any $u\in\R^n$.
\end{definition}

This property yields that $\E\eta_i=\E\eta_j$, but is clearly
weaker than the $ij$-exchangeability of $\eta$.

In the following we often need to change the probability measure
$\Q$. Let $\eta=e^\xi$ and let $\zeta$ be a random variable that
together with $\xi$ builds $(n+1)$-dimensional random vector
$(\xi,\zeta)$. If $e^{\langle w,(\xi,\zeta)\rangle}$ with
$w\in\R^{n+1}$ is integrable, define $\Q^w$ by
\begin{equation}
  \label{eq:change-measure}
  \frac{d\Q^w}{d\Q}=\frac{e^{\langle w,(\xi,\zeta)\rangle}}
  {\E e^{\langle w,(\xi,\zeta)\rangle}}\,,
\end{equation}
i.e. $\Q^w$ is the Esscher transform of $\Q$ with parameter $w$.  In
case $w\in\R^n$, the same notation applies with $w$ extended by zero
component.  If $w=e_j$ is the $j$th standard basis vector in
$\R^{n+1}$, then we write shortly $\Q^j$.
The expectation with respect to changed
measures is indicated by the corresponding subscript.

The following result shows that $ij$-swap-invariance is related to the
self-duality in a lower-dimensional space. Define functions
$\kappat_j:(0,\infty)^n\mapsto(0,\infty)^{n-1}$ acting as
\begin{displaymath}
  \kappat_j(x)=\left(\frac{x_1}{x_j},\dots,\frac{x_{j-1}}{x_j},
    \frac{x_{j+1}}{x_j},\dots,\frac{x_n}{x_j}\right)\,, \quad j=1,\dots,n\,.
\end{displaymath}

\begin{theorem}
  \label{th:swap-inv-self-duality}
  Let $\eta$ be integrable and $i,j\in\{1,\dots,n\}$, $i<
  j$. Then the following two statements are equivalent.
  \begin{itemize}
   \item[(I)]  The $n$-dimensional random vector $\eta$ is $ij$-swap-invariant
     under $\Q$.
   \item[(II)] The $(n-1)$-dimensional random vector $\kappat_j(\eta)$
     is self-dual with respect to the $i$th numeraire under the
     probability measure $\Q^j$.
   \end{itemize}
\end{theorem}
\begin{proof}
  The change of measure formula yields that for all $u\in\R^n$
  \begin{align*}
    \E_{\Q^j}\Big(\sum_{l=1,l\neq j}^n u_l\frac{\eta_l}{\eta_j}+u_j\Big)_+
    &=(\E\eta_j)^{-1}\E \big(\sum_{l=1}^n u_l\eta_l\big)_+\,,\\
    \E_{\Q^j}\Big(\sum_{l=1,l\neq i,j}^n
    u_l\frac{\eta_l}{\eta_j}+u_i+u_j\frac{\eta_i}{\eta_j}\Big)_+
    &=(\E\eta_j)^{-1}\E \big(\sum_{l=1,l\neq i,j}^n
    u_l\eta_l+u_i\eta_j+u_j\eta_i\big)_+\,.
  \end{align*}
  The equality of the right-hand sides characterises the
  $ij$-swap-invariance of $\eta$, while the equality of the left-hand
  sides means the self-duality of $\kappat_j(\eta)$ with respect to
  the $i$th numeraire under $\Q^j$.
\end{proof}

\begin{remark}
  \label{rem:kappat}
  In view of Corollary~\ref{cor:ij-exch} one
  can show by a similar argument that if $n\geq 3$ and $i,j<k$
  (for notational convenience), then condition~(I) holds if and only
  if $\kappat_k(\eta)$ is $ij$-exchangeable under $\Q^k$. In the
  risk-neutral foreign exchange setting $\Q^k$ acquires an immediate
  interpretation in the market where trades take place in the currency
  number $k$.
\end{remark}

\begin{example}[Bivariate swap-invariance and symmetry]
  \label{eg:swap-inv-self-duality}
  Let $\eta$ be a bivariate swap-invariant random vector with
  $\E\eta_1=1$.  Then $\kappat_1(\eta)=\eta_2/\eta_1$ is denoted by
  $\tilde\eta$ and
  \begin{displaymath}
    \E_{\Q^1}(u_1\tilde\eta+u_2)_+=\E (u_1\eta_2+u_2\eta_1)_+
    =\E (u_1\eta_1+u_2\eta_2)_+=\E_{\Q^1}(u_1+u_2\tilde\eta)_+\,,
  \end{displaymath}
  for all $u_1,u_2\in\R$.
  Hence, $(\eta_1,\eta_2)$ is swap-invariant under $\Q$ if and only if
  $\tilde\eta$ satisfies the classical univariate European put-call
  symmetry under the ``dual-market'' measure $\Q^1$. In particular,
  this means that bivariate swap-invariance is not more restrictive
  than the very well-known and often applied European put-call
  symmetry. For the analysis and characterisation of even weaker
  properties we refer to Section~\ref{sec:quasi-swap}.
\end{example}

The expected payoff $\E \fbo(u)$ considered a function of $u$ becomes
the support function of an $n$-dimensional convex body called the
\emph{zonoid} of $\eta$, see e.g.~\cite{mos02} for a detailed
discussion about these well-known convex bodies in relation to random
vectors. In particular, it is well known that zonoids (unlike lift
zonoids from Theorem~\ref{thr:basket}) do not uniquely characterise
the distribution of $\eta$.

The $ij$-swap-invariance of $\eta$ means that its zonoid is symmetric
with respect to the plane $\{u_i=u_j\}$, equivalently, that $\eta$ and
$\pi_{ij}\eta$ share the same zonoid.  In view of this, we first
characterise general non-negative integrable random vectors with equal
zonoids. In the following denote $\onev=(1,\dots,1)$ in the space of
an appropriate dimension.

\begin{theorem}
  \label{th:char-equal-zonoids}
  Let $\eta=e^{\xi}$ and $\eta^*=e^{\xi^*}$ be integrable random
  vectors. Then
  \begin{equation}
   \label{eq:sup-zon}
    \E(\langle u,\eta\rangle)_+=\E(\langle u,\eta^*\rangle)_+\quad\text{for all } u\in\R^n
  \end{equation}
  if and only if
  \begin{equation}
   \label{eq:char-f-zon}
    \varphi_\xi(u-\imagi w)=\varphi_{\xi^*}(u-\imagi w)
  \end{equation}
   for all $u\in\HH$, where
  \begin{equation}
    \label{eq:h}
    \HH  =\{u\in\R^n:\; \sum_{k=1}^n u_k=0\}\,,
  \end{equation}
  and for at least one (and then necessarily for all) $w$,
  such that $\sum w_i=1$ and both sides in (\ref{eq:char-f-zon}) are
  finite.
\end{theorem}
\begin{proof}
  \textsl{Necessity.} Equality (\ref{eq:sup-zon}) implies that
  $\E\eta_i=\E\eta^*_i$ for all $i$. Change measure $\Q$ to $\Q^1$ and
  $\Q^{1*}$ using respectively $\eta_1$ and $\eta_1^*$ as the density
  normalised by the expectation. By Theorem~\ref{thr:basket}, the
  distribution of $\kappat_1(\eta)$ under $\Q^1$ coincides with the
  distribution of $\kappat_1(\eta^*)$ under $\Q^{1*}$. Assume that
  (\ref{eq:char-f-zon}) is finite, i.e.\ $\E e^{\langle
    w,\xi\rangle}<\infty$ for some $w\in\R^n$ with $\sum w_i=1$. Then
  \begin{displaymath}
    f(\xi)=\exp\{\imagi \langle
    (u_2,\dots,u_n)-\imagi (w_2,\dots,w_n),(\xi_2-\xi_1,\dots,\xi_n-\xi_1)\rangle\}
  \end{displaymath}
  is integrable under $\Q^1$, so that $f(\xi^*)$ is integrable under
  $\Q^{1*}$ and both expectations are equal. By changing back
  the measures and using $u=(-\sum_{i=2}^n u_i,u_2,\dots,u_n)\in\HH$ and
  $\sum_{k=1}^n w_k=1$ this implies~(\ref{eq:char-f-zon}).

  \textsl{Sufficiency.} If the both sides of (\ref{eq:char-f-zon}) are
  finite and equal for some $w$, then
  \begin{displaymath}
    \E e^{\langle w,\xi\rangle} = \E e^{\langle w,\xi^*\rangle}=c\,.
  \end{displaymath}
  Thus, the characteristic functions (restricted on $\HH$) of $\xi$
  under the changed measure $\Q^w$ and of $\xi^*$ under $\Q^{w*}$
  coincide, where the change of measure is done with the normalised
  densities $e^{\langle w,\xi\rangle}$ and $e^{\langle
    w,\xi^*\rangle}$ respectively. Therefore, $\xi-\onev \xi_1$ under
  $\Q^w$ is identically distributed as $\xi^*-\onev \xi_1^*$ under
  $\Q^{w*}$. Using that $\sum_{k=1}^n w_k=1$ and changing measures, we obtain
  \begin{align*}
    \E\langle u,\eta\rangle_+&=c\E_{\Q^w}[\langle
    u,e^\xi e^{-\xi_1}\rangle_+e^{\xi_1}e^{-\langle
    w,\xi\rangle}]=c\E_{\Q^w}[\langle
    u,e^{\xi-\onev\xi_1}\rangle_+e^{-\langle w,\xi-\onev\xi_1\rangle}]\,,\\
    \E\langle u,\eta^*\rangle_+&=c\E_{\Q^{w*}}[\langle
    u,e^{\xi^*-\onev\xi^*_1}\rangle_+e^{-\langle
    w,\xi^*-\onev\xi^*_1\rangle}]\,.
  \end{align*}
  Since $\xi-\onev \xi_1$ under
  $\Q^w$ shares the distribution with $\xi^*-\onev \xi_1^*$ under
  $\Q^{w*}$ the right- and thus, also the left hand sides coincide,
  i.e.\ we arrive at~(\ref{eq:sup-zon}). The necessity yields
  that~(\ref{eq:char-f-zon}) holds for all $w$ such that the characteristic
  function is finite and $\sum w_k=1$.
\end{proof}

\begin{remark}
  \label{rem:integr}
  By the generalised H\"older inequality, the integrability of $\eta$
  and $\eta^*$ in Theorem~\ref{th:char-equal-zonoids} yields that the
  characteristic functions in (\ref{eq:char-f-zon}) are finite for all
  $w$ from the unit simplex
  \begin{equation}
    \label{eq:delta}
    \Delta=\{x=(x_1,\dots,x_n):\; x_i\geq 0, i=1,\dots,n,\;
    \sum x_i=1\}\,.
  \end{equation}
  Thus, the set of all $w\in\R^n$ such that $e^{\langle w,\xi\rangle}$
  is integrable contains the unit simplex $\Delta$ if $e^\xi$ is
  integrable itself.
\end{remark}

Let $\sH_\beta$ denote the family of non-negative
positive-$\beta$-homogeneous functions $g:(0,\infty)^n\mapsto\R_+$,
i.e.\ $g(cx)=c^\beta g(x)$ for all $c>0$ and $x\in(0,\infty)^n$. Note
that $\fbo\in\sH_1$. Other examples of payoff functions of class
$\sH_\beta$ can be found in the literature about the duality
principle, see e.g.~\cite{eber:pap:shir08b,faj:mor08}. The following
result says that the equality of expected payoffs from exchange
options implies the equality of expected payoffs from the whole family
$\sH_1$, despite of the fact that the asset prices do not necessarily
coincide in distribution.

\begin{theorem}
  \label{thr:pos-1-hom-gen-two}
  If integrable random vectors $\eta=e^{\xi}$ and $\eta^*=e^{\xi^*}$
  satisfy~(\ref{eq:sup-zon}) (i.e.\ share the same zonoid), then $\E
  g(\eta) =\E g(\eta^*)$ for all $g\in\sH_1$.
\end{theorem}
\begin{proof}
  By choosing $u=e_1$ in (\ref{eq:sup-zon}) we arrive at
  $\E\eta_1=\E\eta^*_1$.  Hence,~(\ref{eq:sup-zon}) is equivalent to
  \begin{equation}
    \label{eq:equal-lz}
    \E_{\Q^1}\Big(u_1+\sum_{i=2}^n u_i\frac{\eta_i}{\eta_1}\Big)_+
    =\E_{\Q^{1*}}\Big(u_1+\sum_{i=2}^n
    u_i\frac{\eta_i^*}{\eta_1^*}\Big)_+
  \end{equation}
  for all $u\in\R^n$. By Theorem~\ref{thr:basket}, the distribution of
  $\kappat_1(\eta)$ under $\Q^1$ coincides with the distribution of
  $\kappat_1(\eta^*)$ under $\Q^{1*}$ so that
  \begin{displaymath}
    \E g(\eta)=c\E_{\Q^1} g((1,\kappat_1(\eta)))
    =c\E_{\Q^{1*}} g((1,\kappat_1(\eta^*)))=\E g(\eta^*)
  \end{displaymath}
  for all $g\in\sH_1$.
\end{proof}

It is possible to generalise this characterisation for functions from
the family $\sH_\beta$.

\begin{theorem}
  \label{thr:beta-hom}
  Assume that random vectors $\eta=e^{\beta\xi}$ and
  $\eta^*=e^{\beta\xi^*}$ are integrable for some $\beta\in\R$. Then
  $\E g(\eta)=\E g(\eta^*)$ for all $g\in\sH_\beta$ if and only if
  \begin{equation}
    \label{eq:beta-hom}
    \phi_{\xi}(u-\imagi w)=\phi_{\xi^*}(u-\imagi w)
  \end{equation}
  for all $u\in\HH$ and at least one (and then necessarily for all)
  $w\in\R^n$ such that $\sum w_k=\beta$ and the characteristic
  functions in (\ref{eq:beta-hom}) exist.
\end{theorem}
\begin{proof}
  Assume first that $\beta\neq0$.  If $g\in\sH_\beta$ with $\beta\neq
  0$, then $g_1(x)=g(x^{1/\beta})$ with the power operation applied
  coordinatewisely is $1$-homogeneous and
  Theorem~\ref{thr:pos-1-hom-gen-two} applies. Thus,
  $\E[g(\eta)]=\E[g(\eta^*)]$ for all $g\in\sH_\beta$ if and only if
  (\ref{eq:char-f-zon}) holds for the characteristic functions of
  $\beta\xi$ and $\beta\xi^*$ (corresponding to $\eta^\beta$ and
  $(\eta^*)^\beta$) for all $u\in\HH$ and at least one (and then
  necessarily for all) $w$ with $\sum w_k=1$. Rewriting
  (\ref{eq:char-f-zon}) for the characteristic functions for $\xi$ and
  $\xi^*$ yields~(\ref{eq:beta-hom}).

  Now let $\beta=0$. For all $u\in\R_+^n$ consider the integrable
  functions $g(\eta)=(u_1-\sum_{i=2}^n u_i \eta_i/\eta_1)_+\in\sH_0$.
  A version of Theorem~\ref{thr:basket} for positive random vectors
  from \cite[Th.~1.1]{hen:sha90} yields that $\kappat_1(\eta)$ and
  $\kappat_1(\eta^*)$ are identically distributed. Calculating the
  characteristic functions of $\xi$ and $\xi^*$ with an arbitrary
  complex argument $w$ yields~(\ref{eq:beta-hom}) with $\sum w_k=0$
  and the characteristic functions exist (at least) for all $w$ with
  vanishing imaginary part.

  In the other direction, (\ref{eq:beta-hom}) implies that
  $\kappat_1(\eta)$ under $\Q^w$ and $\kappat_1(\eta^*)$ under
  $\Q^{w*}$ are identically distributed. Here we have also used that
  (\ref{eq:beta-hom}) yields that $\E e^{\langle w,\xi\rangle}=\E
  e^{\langle w,\xi^*\rangle}=c$. Let $g\in\sH_0$. Then
  $e^{-\langle w,\xi\rangle}=e^{-\langle
      w,\xi-\onev\xi_1\rangle}$ is a function of $(1,\kappat_1(\eta))$,
  so that we can denote
  $f(\kappat_1(\eta))=g((1,\kappat_1(\eta)))e^{-\langle
    w,\xi-\onev\xi_1\rangle}$. Then
  \begin{displaymath}
    \E g(\eta)=c\E_{\Q^w} f(\kappat_1(\eta))=
    c\E_{\Q^{w*}}f(\kappat_1(\eta^*))=\E g(\eta^*)\,.
  \end{displaymath}
\end{proof}

The following result follows directly from
Theorem~\ref{th:char-equal-zonoids} noticing that the
$ij$-swap-invariance of $\eta$ means that $\eta$ and
$\pi_{ij}\eta$ have equal zonoids, or equivalently, that
$\E\langle
u,\eta\rangle_+=\E\langle\pi_{ij}u,\eta\rangle_+=\E\langle
u,\pi_{ij}\eta\rangle_+$ for all $u\in\R^n$.

\begin{corollary}
  \label{th:char-swap-n}
  An integrable random vector $\eta=e^\xi$ is $ij$-swap-invariant if
  and only if the characteristic function of $\xi$ satisfies
  \begin{equation}
    \label{eq:char-swap-inv}
      \varphi_\xi(u-\imagi w)
      =\varphi_\xi(\pi_{ij}(u-\imagi w))
  \end{equation}
  for all $u\in\HH$ and for at least one (and then necessarily for
  all) $w\in \Delta$.
\end{corollary}

While the characteristic function (\ref{eq:char-swap-inv}) exists for
all $w\in\Delta$, it is possible to relax the latter condition.
Namely, integrable $\eta$ is $ij$-swap-invariant if and only
if~(\ref{eq:char-swap-inv}) holds for all $u\in\HH$ and at least one
(and then necessarily for all) vectors $w$, such that
$\sum_{k=1}^nw_k=1$ and one side of~(\ref{eq:char-swap-inv}) is
finite, in other words such that $\E e^{\langle w,\xi\rangle}<\infty$.

The complex shifts on both sides of (\ref{eq:char-swap-inv}) are the
same if $w_i=w_j$. In the most important special case $w=\thf e_{ij}$
with $e_{ij}=e_i+e_j$, so that the $ij$-swap-invariance
characterisation reads
\begin{equation}
  \label{eq:half}
  \varphi_\xi(u-\imagi\thf e_{ij})
  =\varphi_\xi(\pi_{ij}u-\imagi\thf e_{ij})\,,\quad u\in\HH\,.
\end{equation}

\begin{corollary}
  \label{cor:projected-xi}
  An integrable random vector $\eta=e^\xi$ is $ij$-swap-invariant if
  and only if the orthogonal projection
  of $\xi$ onto $\HH$ is $ij$-exchangeable under the
  probability measure $\Q^w$ for at least one (and then necessarily
  for all) $w\in\Delta$ such that $w_i=w_j$.
\end{corollary}

\begin{remark}[Independency and self-duality in the bivariate case]
  Consider a bivariate integrable random vector $\eta=(\eta_1,\eta_2)$
  with independent components. A sufficient condition for $\eta$ to be
  swap-invariant is that both $\eta_1$ and $\eta_2$ are self-dual
  random variables, since then we have for arbitrary
  $(u_1,u_2)\in\R^2$
  \begin{align*}
    \E[(u_1\eta_1+u_2\eta_2)_+]&=\E[\E[(u_1\eta_1+u_2\eta_2)_+|\eta_1]]\\
    &=\E[\E[(u_1\eta_1\eta_2+u_2)_+|\eta_2]]
    =\E[(u_1\eta_2+u_2\eta_1)_+]\,.
  \end{align*}
  Note that this construction does not apply for $\eta$ of dimension
  $3$ and more.  Non-exchangeable swap-invariant random vectors with
  independent not necessarily self-dual components can be constructed
  in the following way.  Consider integrable i.i.d.\ $\zeta_1,\zeta_2$
  and self-dual $\tilde\zeta_1,\tilde\zeta_2$ all jointly independent.
  It is easy to see that
  $(\eta_1,\eta_2)=(\zeta_1\tilde\zeta_1,\zeta_2\tilde\zeta_2)$ is a
  swap-invariant random vector with not necessarily self-dual
  independent components.
  It is apparent from \cite{mol:sch10} that the product of a self-dual
  random variable and a general one is not necessarily self-dual.
\end{remark}

\section{Weighted swap-invariance}
\label{sec:weight-swap-invar}

The introduced swap-invariance concept relies on invariance
properties of payoff function $\fbo$ from (\ref{eq:fbo}). It is
also possible to modify this payoff function by introducing a
positive weight given by a random variable $e^\vartheta$.  A
random vector $\eta$ is called \emph{weighted $ij$-swap-invariant}
if $e^{\vartheta}\eta$ is integrable and
\begin{equation}
  \label{eq:weighted-prop}
  \E(e^{\vartheta}\fbo(u))=\E(e^{\vartheta}\fbo(\pi_{ij}(u)))\quad\text{for all}\; u\in\R^n\,.
\end{equation}
In this case we write $\eta\in\WSI_{ij}(\vartheta)$.
The involved payoff function is typical for so-called
\emph{quanto}-swap options.

\begin{theorem}
  \label{th:weighted-char-swap}
  Let $\eta=e^\xi$ be a random vector and let $e^\vartheta$ be a random
  variable such that $e^{\vartheta}\eta$ is integrable.  Then
  $\eta\in\WSI_{ij}(\vartheta)$ if and only if
  \begin{equation}
    \label{eq:weighted-on-char}
    \phi_{\xi+\onev\vartheta}(u-\imagi w)
    =\phi_{\xi+\onev\vartheta}(\pi_{ij}(u-\imagi w))
  \end{equation}
  for all $u\in\HH$ and for at least one (and then necessarily for
  all) $w\in\Delta$.
\end{theorem}
\begin{proof}
  It suffices to note that $e^\vartheta\eta$ and $e^\vartheta\pi_{ij}\eta$
  share the same zonoid and apply Theorem~\ref{th:char-equal-zonoids}.
\end{proof}

If $w=\thf e_{ij}$, then (\ref{eq:weighted-on-char}) simplifies to
\begin{equation}
  \label{eq:char-gen-quanto-inv-w-char-f}
  \varphi_{\xi+\onev\vartheta}(u-\imagi \thf e_{ij})
  =\varphi_{\xi+\onev\vartheta}(\pi_{ij}u-\imagi\thf
  e_{ij})\quad\text{for all }u\in\HH\,.
\end{equation}

If the log-weight $\vartheta$ is given by a linear combination of
the log-prices of the assets included in $\fbo(u)$, i.e.\
$\vartheta=\langle\xi,v\rangle$ for some $v\in\R^n$, we obtain the
following result.

\begin{corollary}
 \label{cor:lin-com-weight}
 Let $\eta=e^\xi$ be a random vector such that $e^{\vartheta}\eta$ is
 integrable with $\vartheta=\langle v,\xi\rangle$ for some $v\in\R^n$.
 Then $\eta\in\WSI_{ij}(\vartheta)$ if and only if
 \begin{equation}
    \label{eq:char-gen-quanto-inv-w-char-f-lin}
      \varphi_\xi(u-\imagi w - \imagi v)
      =\varphi_\xi(\pi_{ij}(u-\imagi w) - \imagi v)
  \end{equation}
  for all $u\in\HH$ and for at least one (and then necessarily for
  all) $w\in\Delta$.
\end{corollary}
\begin{proof}
  Since $u\in\HH$ and $w\in\Delta$,
  \begin{align*}
    \varphi_{\xi+\onev\vartheta}(u-\imagi w)&=\E e^{\imagi(\langle
      u-\imagi w,\xi+\onev\langle
      v,\xi\rangle\rangle)}\\
    &=\E e^{\imagi\langle u-\imagi w,\xi\rangle+\imagi\langle
      v,\xi\rangle
      (\langle u,\onev\rangle-\imagi\langle w,\onev\rangle)}
    =\E e^{\imagi\langle u-\imagi w-\imagi v,\xi\rangle}
  \end{align*}
  and analogously $\varphi_{\xi+\onev\vartheta}(\pi_{ij}(u-\imagi w))
  =\varphi_\xi(\pi_{ij}(u-\imagi w)- \imagi v)$.
\end{proof}

\section{Swap-invariance for L\'evy models}
\label{sec:swap-invar-infin}

In this section we assume that $\eta=e^\xi$ with $\xi$ being
infinitely divisible, i.e.\ $\xi=L_1$ for a L\'evy process $L_t$,
$t\geq0$, see~\cite{sat99}. In order to handle possibly weighted cases
consider also a random variable $\zeta$ such that $(\xi,\zeta)$ is
infinitely divisible.

Define the linear transformation (actually orthogonal projection),
which maps every $x\in\R^{n+1}$ onto the hyperplane $\HH$ in the space
of dimension $n$ acting as $Px$ with the matrix
\begin{equation}
  \label{eq:P}
  P=
  \begin{pmatrix}
    1-\frac{1}{n} & -\frac{1}{n} & \cdots & -\frac{1}{n} & 0\\[2mm]
    -\frac{1}{n} & 1-\frac{1}{n} & \cdots & -\frac{1}{n} & 0\\[2mm]
    \vdots & \vdots & \vdots & \vdots & \vdots \\
    -\frac{1}{n} & -\frac{1}{n} & \cdots & 1-\frac{1}{n} & 0
  \end{pmatrix}\,.
\end{equation}

Corollary~\ref{th:char-swap-n},
Theorem~\ref{th:weighted-char-swap}, and
Corollary~\ref{cor:lin-com-weight} provide many equivalent
characterisations of the (weighted) $ij$-swap-invariance in terms
of various $w\in\Delta\subset\R^n$. In order to simplify the
calculations we let $w=\thf(e_i+e_j)$ in the sequel (so that $w$
is $\pi_{ij}$-invariant) and we can consider $e_i$ and $e_j$ to be
standard basis vectors in $\R^{n+1}$. Sometimes we add the zero
component to the vectors $u\in\HH$ and then write $(u,0)$.

\begin{theorem}
  \label{th:inv-div-prime}
  Let $(\xi,\zeta)$ be infinitely divisible $(n+1)$-dimensional random
  vector such that $e^{\langle v,(\xi,\zeta)\rangle}e^\xi$ is
  integrable for some $v\in\R^{n+1}$.  Then
  $\eta=e^\xi\in\WSI_{ij}(\langle v,(\xi,\zeta)\rangle)$ if and only
  if the characteristic triplet $(A,\nu,\gamma)$ of $(\xi,\zeta)$
  satisfies the following conditions.
   \begin{itemize}
   \item[(1)] If $n\geq3$, the matrix $A$ satisfies
     \begin{equation}
       \label{eq:asi}
       a_{li}-a_{lj}=\thf(a_{ii}-a_{jj})
     \end{equation}
     for all $l\neq i,j$, $l\leq n$.
   \item[(2)] The image $\hat\nu P^{-1}$ under $P$ of measure
     \begin{equation}
       \label{eq:nu-Essc}
       d\hat\nu(x)=e^{\langle \thf e_{ij}+v,x\rangle}d\nu(x)\,,\quad x\in\R^{n+1}\,,
     \end{equation}
     is $\pi_{ij}$-invariant on $\HH\setminus\{0\}$.
  \item[(3)] $\gamma$ satisfies
    \begin{multline}
      \label{eq:th-third}
      \gamma_i-\gamma_j=\thf(a_{jj}-a_{ii})+\sum_{k=1}^{n+1}
      (a_{jk}-a_{ik})v_k\\
      +\int_{\R^{n+1}}(x_j-x_i)(e^{\langle \thf e_{ij}+v,x\rangle}
      \one_{\|Px\|\leq 1}-\one_{\|x\|\leq 1})d\nu(x)\,.
    \end{multline}
  \end{itemize}
\end{theorem}
\begin{proof}
  Since $u\in\HH$ we can express~(\ref{eq:char-gen-quanto-inv-w-char-f}) in terms
  of the joint characteristic function of $(\xi,\zeta)$, i.e.
  \begin{equation}
    \label{eq:char-gen-quanto-inv-w-char-f-ref}
    \varphi_{(\xi,\zeta)}((u,0)-\imagi(\thf e_{ij}+v))=
    \varphi_{(\xi,\zeta)}(\pi_{ij}(u,0)-\imagi(\thf
    e_{ij}+v))\,,
  \end{equation}
  where from now on $\pi_{ij}$ stands for the corresponding permutation matrix
  of appropriate dimension. Let $\hat\Q$ be the Esscher transform~(\ref{eq:change-measure})
  of $\Q$ with parameter $w=\thf e_{ij}+v$. The characteristic triplet
  $(\hat A,\hat\nu,\hat\gamma)$ of $(\xi,\zeta)$ under $\hat\Q$ is given by
  $\hat A=A$, the new L\'evy measure $\hat\nu$ is given by (\ref{eq:nu-Essc}),
  and
  \begin{equation}
    \label{eq:ga-Essc}
    \hat\gamma=\gamma+A(\thf e_{ij}+v)
    +\int_{\R^{n+1}}x(e^{\langle \thf e_{ij}+v,x\rangle}-1)
    \one_{\|x\|\leq 1} d\nu(x)\,,
  \end{equation}
  see~\cite[Ex.~7.3]{sat00} and~\cite[Th.~25.17]{sat99} for the
  extension of the L\'evy-Khintchine formula to the needed
  subset in the $(n+1)$-dimensional complex plane in view of the imposed
  integrability conditions.

  By~(\ref{eq:char-gen-quanto-inv-w-char-f-ref}) and in view of
  Corollary~\ref{cor:projected-xi}, the weighted swap-invariance of
  $\eta$ means that $(\xi,\zeta)$ projected by $P$ onto $\HH$ is
  $ij$-exchangeable under $\hat\Q$.  This projection has the
  characteristic triplet $(A',\nu',\gamma')$, where
  $A'=P\hat{A}P^\top=PAP^\top$, $\nu'=\hat\nu P^{-1}$ is the projection of $\hat\nu$
  on $\HH\setminus\{0\}$ and
  \begin{equation}
    \label{eq:ga-proj}
    \gamma'=P\hat{\gamma}+\int_{\R^{n+1}} Px(\one_{\|Px\|\leq 1}-\one_{\|x\|\leq1})d\hat{\nu}(x)\,,
  \end{equation}
  see~\cite[Prop.~11.10]{sat99}. The elements of $A'$ can be calculated as
  \begin{displaymath}
    a'_{ij}=a_{ij}-\frac{1}{n}
    \Big(\sum_{k=1}^n a_{ki}+\sum_{k=1}^n a_{kj}\Big)+\frac{1}{n^2}\sum_{k,l=1}^na_{lk}\,.
  \end{displaymath}
  Since the projection of $(\xi,\zeta)$ is $ij$-exchangeable,
  Proposition~\ref{th:inv-div-CE}(1) requires $a'_{ii}=a'_{jj}$,
  so that
  \begin{equation}
    \label{eq:asilong}
    a_{ii}-\frac{2}{n}\sum_{k=1}^n a_{ik}
    =a_{jj}-\frac{2}{n}\sum_{k=1}^n a_{jk}\,.
  \end{equation}
  Furthermore, $a'_{li}=a'_{lj}$ for $l\neq i,j$ yields
  (\ref{eq:asi}), which also always satisfies (\ref{eq:asilong}).  By
  Proposition~\ref{th:inv-div-CE}, $\nu'$ is symmetric with respect to
  $\pi_{ij}$ and $\gamma'_i=\gamma'_j$. By
  combining~(\ref{eq:ga-proj}) with~(\ref{eq:ga-Essc}) we obtain
  (\ref{eq:th-third}).
\end{proof}

By combining Theorem~\ref{th:inv-div-prime}
with~\cite[Prop.~11.10]{sat99} and changing variables, or adapting
the proof of Theorem~\ref{th:inv-div-prime}, we obtain the
following result.

\begin{corollary}
  \label{cor:non-weighted}
  The integrable random vector $\eta=e^\xi$ with infinitely divisible $\xi$
  having the L\'evy triplet $(A,\nu,\gamma)$ is $ij$-swap-invariant if
  and only if condition (1) of Theorem~\ref{th:inv-div-prime} holds for
  the $n\times n$ matrix $A$, the orthogonal projection of measure
  \begin{equation}
    \label{eq:bar-nu}
    d\bar\nu(x)=e^{\thf(x_i+x_j)}d\nu(x)\,,\quad
    x\in\R^n\,,
  \end{equation}
  on $\HH\setminus\{0\}$ is $\pi_{ij}$-invariant
  and
  \begin{equation}
    \label{eq:third-no-zeta}
    \gamma_i-\gamma_j=\thf(a_{jj}-a_{ii})
    +\int_{\R^n}(x_j-x_i)(e^{\thf(x_i+x_j)}
    \one_{\|P'x\|\leq 1}-\one_{\|x\|\leq 1})d\nu(x)\,,
  \end{equation}
  where $P'$ is $P$ with the last column omitted.
\end{corollary}

The following theorem shows that the condition on the drift $\gamma$
from the L\'evy triplet is automatically satisfied in case of equal
means.

\begin{theorem}
  \label{thr:rnc}
  Let $\eta=e^\xi$ be an $n$-dimensional integrable random vector with
  infinitely divisible $\xi$ and such that $\E\eta_i=\E\eta_j$.  Then
  $\eta$ is $ij$-swap-invariant if and only if the characteristic
  triplet $(A,\nu,\gamma)$ of $\xi$ satisfies the first two conditions
  of Corollary~\ref{cor:non-weighted} (i.e.~(\ref{eq:third-no-zeta})
  always holds in this case).
\end{theorem}
\begin{proof}
  Since $\E\eta_l=\phi_\xi^\Q(-\imagi e_l)$ for $l=i,j$ (where here $e_l\in\R^n$),
  \begin{multline*}
    \gammaQ_i+\thf a_{ii}+\int_{\R^n}(e^{x_i}-1-x_i\one_{\| x\|\leq 1})
    d\nuQ(x)\\=\gammaQ_j+\thf a_{jj}+\int_{\R^n}(e^{x_j}-1-x_j\one_{\| x\|\leq 1})
    d\nuQ(x)\,.
  \end{multline*}
  In this case (\ref{eq:third-no-zeta}) turns into
  \begin{displaymath}
    \int_{\R^n}(e^{\thf(x_i-x_j)}-e^{\thf(x_j-x_i)}
    +(x_j-x_i)\one_{\|P'x\|\leq 1})d\bar\nu(x)=0\,.
  \end{displaymath}
  Changing variable as $x=x'+x''$ with $x'\in\HH$ and
  $x''\in\HH^\perp$ and noticing that $\HH^\perp=\{t\onev:\; t\in\R\}$
  consists of vectors with all equal components, the integral turns
  into an integral over $\HH$ with respect to the projection of
  $\bar\nu$ onto $\HH$. The integrand changes the sign if $x$ is
  replaced by $\pi_{ij}x$, while the projected measure $\bar\nu$ is
  invariant on $\HH\setminus\{0\}$ (where the integrand is non-vanishing)
   under this change. Thus, the whole integral vanishes.
\end{proof}

\begin{remark}[Risk-neutral non-weighted case]
  It is worth noticing that the assumption $\E\eta_i=\E\eta_j$ in
  Theorem~\ref{thr:rnc} is satisfied in a risk-neutral setting, where
  $\E\eta_l=1$, $l=1,\dots,n$.
\end{remark}

\begin{example}[Two-asset case]
  \label{cor:twod}
  In the bivariate non-weighted infinitely divisible (L\'evy) case the
  first condition of Corollary~\ref{cor:non-weighted} is vacuous.  The
  second condition holds, e.g.\ for exchangeable $\nu$, while the
  third one always holds in the risk-neutral setting.
\end{example}

\begin{example}[Log-normal distribution]
  \label{rem:gbm}
  If the L\'evy measure vanishes, the first condition of
  Theorem~\ref{th:inv-div-prime} remains the same, the second
  condition always holds, while the third one becomes (with $\mu$
  written instead of $\gamma$)
  \begin{displaymath}
    \mu_i-\mu_j=\thf(a_{jj}-a_{ii})+\sum_{k=1}^{n+1}(a_{jk}-a_{ik})v_k\,.
  \end{displaymath}
  Under a risk-neutral assumption this condition means that
  $\sum_{k=1}^{n+1}(a_{jk}-a_{ik})v_k=0$, so in the non-weighted risk-neutral
  setting only the first condition of Corollary~\ref{cor:non-weighted}
  is imposed.

  In particular, \emph{each bivariate} risk-neutral
  log-normal distribution is (non-weighted) swap-invariant, no matter
  what volatilities of the assets and correlation are. In the
  non-weighted risk-neutral setting with $n=3$ and $i=1$, $j=2$ the only
  condition is
  \begin{displaymath}
    a_{31}-a_{32}=\thf(a_{11}-a_{22})\,.
  \end{displaymath}

  In the presence of a weight $(\xi_1,\xi_2,\zeta)$, i.e.\ $\vartheta=\zeta$,
  in the risk-neutral case the only condition $a_{13}=a_{23}$ on the
  covariance matrix of $(\xi_1,\xi_2,\zeta)$ guarantees the weighted
  swap-invariance property
  \begin{displaymath}
    \E e^{\zeta}(u_1\eta_1+u_2\eta_2)_+
    =\E e^{\zeta}(u_2\eta_1+u_1\eta_2)_+,\quad (u_1,u_2)\in\R^2\,.
  \end{displaymath}
  If the weight is determined by the prices of the assets included in
  the swap, namely $\vartheta=\langle v,\xi\rangle$, the
  swap-invariance condition reads
  $v_1a_{11}-v_2a_{22}=(v_1-v_2)a_{12}$.

  Consider a higher dimensional risk-neutral log-normal setting with
  the weight $\vartheta=\langle v,\xi\rangle$ determined by the assets
  included in $\fbo(u)$ in a rather general way with $v\notin\HH$ and
  $v_i=v_j$ (in particular, $v=e_k$ with $k\neq i,j$). Then
  $\eta\in\WSI_{ij}(\langle v,\xi\rangle)$ implies the
  $ij$-exchangeability of $\eta$. Indeed, the risk-neutrality reduces
  (\ref{eq:th-third}) to $\sum_{k=1}^n(a_{jk}-a_{ik})v_k=0$. By
  $v_i=v_j$ and~(\ref{eq:asi}) this yields that
  $\thf(a_{ii}-a_{jj})\langle\onev,v\rangle=0$, whence $a_{ii}=a_{jj}$
  by $v\notin\HH$.  Taking into account~(\ref{eq:asi}), we have also
  $a_{li}=a_{lj}$ for all $l\neq i,j$, so that the exchangeability
  follows from Proposition~\ref{th:inv-div-CE}.
\end{example}

\begin{remark}[Square integrable case and covariance]
  \label{re:korr}
  Condition~(1) in Theorem~\ref{th:inv-div-prime} yields a certain
  restriction on the correlation structure arising from the centred
  Gaussian term for $n\geq 3$, while for $n=2$ there are no restrictions.
  In order to relax the restrictions also for higher-dimensional
  models, it is useful to introduce a jump component. Assume that
  $\int_{\|x\|>1}\|x\|^2d\nu(x)<\infty$, i.e.\ $(\xi,\zeta)$
  is square-integrable. Then the covariance matrix of
  $(\xi,\zeta)$ has elements
  \begin{displaymath}
    \Sigma_{lj}=\left(a_{lj}+\int x_lx_jd\nu(x)\right)\,,\quad
    l,j=1,\dots,n+1\,,
  \end{displaymath}
  see~\cite[Ex.~25.12]{sat99}. Thus, despite of some
  constrains on the L\'evy measure given in
  Theorem~\ref{th:inv-div-prime} (resp.\
  Corollary~\ref{cor:non-weighted}), there is more
  flexibility in modelling the correlation structure of $\eta$.
\end{remark}

\begin{remark}[L\'evy measures based on exchangeability]
  \label{rem:lm}
  The image of the L\'evy measure $\hat\nu$ under $P$ (resp.\ the
  projection of $\bar\nu$ on $\HH\setminus\{0\}$) is $\pi_{ij}$-invariant if (but
  not only if) the L\'evy measure $\hat\nu$ (resp.\ $\bar\nu$) is
  $\pi_{ij}$-invariant itself. Simple example of L\'evy measures
  satisfying Theorem~\ref{th:inv-div-prime}(2) (resp.\
  Corollary~\ref{cor:non-weighted}(2)) can be constructed by
  taking an $(n+1)$-dimensional (resp.\ $n$-dimensional)
  $ij$-exchangeable (i.e.\ $\pi_{ij}$-invariant) L\'evy measure $\tilde\nu$
  (resp.\ $\bar\nu$) satisfying (\ref{eq:nu-conditions}) and defining
  $\nu$ from (\ref{eq:nu-Essc}) or (\ref{eq:bar-nu}) given that the
  imposed integrability assumptions on $e^{\langle v,(\xi,\zeta)\rangle}\eta$ and
  $\eta$ are satisfied.
\end{remark}

\begin{example}[Compound Poisson distribution]
  Assume that the L\'evy measure is finite with existing first
  exponential moments. Without loss of generality assume that its
  total mass is one. Then $\bar\nu$ from
  Corollary~\ref{cor:non-weighted} is, up to a constant, the Esscher
  transform of $\nu$ with parameter $\thf e_{ij}$. Thus, the invariance
  of its projection onto $\HH$ is equivalent to
  \begin{displaymath}
    \varphi_\nu(u-\imagi\thf
    e_{ij})=\varphi_\nu(\pi_{ij}u-\imagi\thf e_{ij})\,,\quad
    u\in\HH\,,
  \end{displaymath}
  for the characteristic function of $\nu$, which exactly corresponds
  to~(\ref{eq:half}).  Hence, the distribution of the logarithm of any
  $ij$-swap-invariant vector $\eta$ can be chosen to serve as the
  L\'evy measure $\nu$ (where $\nu(\{0\})$ is set to zero if $\eta$
  has an atom at $(1,\dots,1)$).  For instance, L\'evy measures
  satisfying~(\ref{eq:bar-nu}) can be created from normal
  distributions described in Example~\ref{rem:gbm}.
  In the bivariate case this imposes only a slight restriction on the
  expectations, while the variances and correlation are not
  restricted.
\end{example}

\begin{example}[Swap-invariance in bivariate generalized hyperbolic
  models]
  Consider a risk-neutral bivariate generalised hyperbolic case, i.e.\
  $\eta=e^{\xi}$, where $(\xi_1,\xi_2)\sim\mathrm{GH}_2(\lambda,\alpha,\beta,\delta,\mu,\Delta)$,
  cf.~\cite{bar77}, with corresponding parameters $\lambda\in\R$,
  $\alpha,\delta\in\R_+$, $\mu,\beta\in\R^2$, and $\Delta$ is a
  symmetric, positive definite, $2\times 2$ matrix, where w.l.o.g. $\det(\Delta)=1$.
  Following~\cite[Ex.~5.9]{eber:pap:shir08b} based on~\cite{mas04}
  assume that $\delta>0$ and $\alpha^2-\langle\beta,\Delta\beta\rangle>0$ so that
  the moments of all orders exist and the L\'evy measure $\nu$ has a density $\nu(x)$ given by
  \begin{multline*}
    \nu(x)=\frac{e^{\langle\beta,x\rangle}}{\pi\sqrt{\langle
    x,\Delta^{-1}x\rangle}}\Big(\int_0^\infty\frac{\sqrt{2y+\alpha^2}
    K_1(\sqrt{(2y+\alpha^2)\langle
    x,\Delta^{-1}x\rangle})}{\pi^2y(J^2_{|\lambda|}(\delta\sqrt{2y})
    +Y^2_{|\lambda|}(\delta\sqrt{2y}))}\,dy\\
    +\alpha K_1(\alpha\sqrt{\langle x,\Delta^{-1}
    x\rangle})\lambda\one_{\{\lambda>0\}}\Big)\,,
  \end{multline*}
  where $J_\iota$, $Y_\iota$ and $K_\iota$ denote the (modified)
  Bessel functions of first, second and third kind with index
  $\iota$, 
  and where further conditions on the parameters for
  ensuring the existence of the exponential moments can
  immediately be obtained from~\cite[Rem.~2.2]{ham04}.

  The parameters for $\tilde\xi=\xi_2-\xi_1$ under $\Q^1$ are
  calculated in~\cite[Ex.~5.9]{eber:pap:shir08b}, in particular
  \begin{displaymath}
    \tilde\beta=\frac{\beta_2\delta_{22}-(\beta_1+1)
    \delta_{11}-\delta_{12}(\beta_2-\beta_1-1)}
    {(\delta_{11}+\delta_{22}-2\delta_{12})}\,.
  \end{displaymath}
  By Theorem~\ref{th:swap-inv-self-duality} or Example~\ref{eg:swap-inv-self-duality}
  $(e^{\xi_1},e^{\xi_2})$ is swap-invariant if and only if $e^{\tilde\xi}$ is
  self-dual under $\Q^1$. However, following~\cite{faj:mor06b}
  in the risk-neutral setting this is the case if and only if
  $\tilde\beta=-\thf$ so that we obtain the slight restriction
  \begin{displaymath}
    2(\delta_{22}-\delta_{12})\beta_2+\delta_{22}
    =2(\delta_{11}-\delta_{12})\beta_1+\delta_{11}\,.
  \end{displaymath}
  Hence, the considerable effective degrees of freedom for modelling
  two assets based on the considered dependent generalised hyperbolic
  L\'evy processes only slightly diminishes by extra imposing 
  the bivariate swap-invariance property holds.
\end{example}

By interpreting $(\xi,\zeta)$ (resp.\ $\xi$) as time one value of
a L\'evy process we arrive at the following result.

\begin{corollary}
  \label{cor:levy-all-T-gen-quant-inv}
  If $(\xi_t,\zeta_t)$ (resp.\ $\xi_t$), $t\geq0$, is the L\'evy process
  with generating triplet $(A,\nu,\gamma)$ that satisfies the conditions of
  Theorem~\ref{th:inv-div-prime} (resp.\ Corollary~\ref{cor:non-weighted}),
  then $e^{\xi_t}$ is weighted $ij$-swap-invariant for all $t\geq0$.
\end{corollary}

\begin{remark}[Random times]
  \label{rem:tau-T-mult}
  Consider a family $\{\eta(t),\,t\geq0\}$ of $ij$-swap-invariant
  random vectors. Let $\tau_t$, $t\geq0$, be an increasing
  non-negative random function independent of $\eta$. If the
  time-changed stochastic process $\eta(\tau_t)$, $t\geq0$, is
  integrable for all $t$, then $\eta(\tau_t)$ is also
  $ij$-swap-invariant.
\end{remark}

\section{Quasi-swap-invariance}
\label{sec:quasi-swap}

In some cases the swap-invariance condition is too restrictive, in
particular, its relaxed variant is useful to adjust for
\emph{unequal carrying costs}. We say that $\eta$ is
\emph{quasi-swap-invariant} if
\begin{equation}
  \label{eq:qsi}
  \E [e^\vartheta\fbo(u)]
  =\E\left[e^\vartheta\fbo(\pi_{ij}(u))(\frac{\eta_i}{\eta_j})^\alpha\right]
\end{equation}
for all $u\in\R^n$ and all mentioned expectations exist. Note that
this property is not symmetric with respect to $i$ and $j$.

By passing to the new probability measure $\tilde\Q^j$ defined
by~(\ref{eq:change-measure}) with $w=e_j+e_{n+1}$, for
$\vartheta=\zeta$, respectively with $w=e_j+\langle
v,(\xi,\zeta)\rangle$ for $\vartheta=\langle
v,(\xi,\zeta)\rangle$,
assuming the $\tilde\Q^j$-integrability of $\tilde\kappa_j(\eta)$ as
well as $\kappat_j(\eta)^{\alpha+1}$ and
using~\cite[Th.~5.2]{mol:sch10} (with vanishing
$\lambda$) 
it is easy to see that~(\ref{eq:qsi}) with $\alpha\neq -1$ is
equivalent to the fact that $\kappat_j(\eta)^{\alpha+1}$ is self-dual
with respect to the $i$th numeraire under $\tilde\Q^j$.  Random
vectors that become self-dual if normalised and raised to some power
are called \emph{quasi-self-dual} in~\cite{mol:sch10}.

\begin{theorem}
  \label{th:qusi-char-on-chaf-f}
  Let $\eta=e^\xi$ be a random vector such that $e^{\vartheta}\eta$ and
  $e^{\vartheta}(\eta_i/\eta_j)^\alpha \eta$ are integrable.
  Then~(\ref{eq:qsi}) holds if and only if
  \begin{equation}
    \label{eq:quasi-v}
    \varphi_{\xi+\onev\vartheta}(u-\imagi w)
    =\varphi_{\xi+\onev\vartheta}(\pi_{ij}(u-\imagi
    w)-\imagi\alpha(e_i-e_j))
  \end{equation}
  for all $u\in\HH$ and for at least one (and then necessarily for all)
  $w\in\Delta$.
\end{theorem}
\begin{proof}
  Define $\vartheta'=\vartheta+\alpha\xi_i-\alpha\xi_j$ and note that
  $e^{\vartheta}\eta$ and $e^{\vartheta'}\eta$ are integrable. Then
  (\ref{eq:qsi}) means that $e^\vartheta\eta$ and $\pi_{ij}e^{\vartheta'}\eta$ share
  the same zonoid.  By Theorem~\ref{th:char-equal-zonoids}, this holds
  if and only if
  \begin{displaymath}
    \phi_{\xi+\onev\vartheta}(u-\imagi w)
    =\phi_{\pi_{ij}(\xi+\onev\vartheta')}(u-\imagi w)
  \end{displaymath}
  for all $u\in\HH$ and for at least one (and then necessarily for
  all) $w\in\Delta$. It remains to note that the right-hand side
  is
  \begin{displaymath}
    \phi_{\xi+\onev\vartheta'}(\pi_{ij}(u-\imagi w))
    =\phi_{\xi+\onev\vartheta}(\pi_{ij}(u-\imagi w)-\imagi\alpha(e_i-e_j))\,.
  \end{displaymath}
\end{proof}

Since $e^\vartheta\eta$ and $\pi_{ij}e^{\vartheta'}\eta$ share the
same zonoid, Theorem~\ref{thr:pos-1-hom-gen-two} yields that
(\ref{eq:qsi}) implies
\begin{equation}
  \label{eq:pos-quasi}
  \E[e^{\vartheta}g(\eta)]
  =\E\left[e^\vartheta
    g(\pi_{ij}(\eta))(\frac{\eta_i}{\eta_j})^\alpha\right]
\end{equation}
for all $g\in\sH_1$.

The assumption $w\in\Delta$ in Theorem~\ref{th:qusi-char-on-chaf-f}
can be replaced by assuming that $\sum w_k=1$ and at least one side of
(\ref{eq:quasi-v}) is finite. Note that the integrability of
$e^{\vartheta}\eta$ and $e^{\vartheta}(\eta_i/\eta_j)^\alpha \eta$
implies $\E
e^{\thf(1+\alpha)\xi_i+\thf(1-\alpha)\xi_j+\vartheta}<\infty$.  Hence,
we can choose $w=\thf(1+\alpha)e_i+\thf(1-\alpha)e_j$, so that
(\ref{eq:quasi-v}) turns into
\begin{multline}
  \label{eq:quasi-v-special-w}
  \varphi_{\xi+\onev\vartheta}(u-\imagi[\thf[1+\alpha]e_i+\thf[1-\alpha]e_j])\\
  =\varphi_{\xi+\onev\vartheta}(\pi_{ij}u-\imagi[\thf[1+\alpha]e_i+\thf[1-\alpha]e_j])
\end{multline}
for all $u\in\HH$, i.e.\ the complex shifts on both sides of
(\ref{eq:quasi-v-special-w}) are identical. For $\vartheta=\langle
v,(\xi,\zeta)\rangle$ we can use that $u\in\HH$ in order to
express~(\ref{eq:quasi-v-special-w}) in terms of the joint
characteristic function of $(\xi,\zeta)$ as
\begin{multline}
  \label{eq:char-gen-quanto-inv-w-char-f-ref-quasi}
  \varphi_{(\xi,\zeta)}((u,0)-\imagi[\thf e_{ij}+\thf\alpha(e_i-e_j)+v])\\=
  \varphi_{(\xi,\zeta)}(\pi_{ij}(u,0)-\imagi[\thf e_{ij}+\thf\alpha(e_i-e_j)+v])
\end{multline}
for all $u\in\HH$.
Hence,~(\ref{eq:char-gen-quanto-inv-w-char-f-ref-quasi})
corresponds to~(\ref{eq:char-gen-quanto-inv-w-char-f-ref}) written for
\begin{displaymath}
  v'=v+\frac{\alpha}{2}(e_i-e_j)
\end{displaymath}
instead of $v$. Thus, in the infinite divisible case under suitable
integrability assumptions, the quasi-swap-invariance holds if and only
if conditions of Theorem~\ref{th:inv-div-prime} are satisfied with $v$
replaced by $v'$ given above, so that we immediately obtain the
following result.

\begin{corollary}
  \label{cor:quasi-levy}
  Let $\eta=e^\xi$ be a random vector such that $e^{\langle v,(\xi,\zeta)\rangle}\eta$
  and $e^{\langle v,(\xi,\zeta)\rangle}(\eta_i/\eta_j)^\alpha \eta$ are integrable
  for some $v\in\R^{n+1}$ and with $(\xi,\zeta)$ being infinitely
  divisible. Then $e^\xi$ is quasi-swap-invariant of order $\alpha$ (i.e.\
  satisfies~(\ref{eq:qsi}) with $\vartheta=\langle v,(\xi,\zeta)\rangle$) if and only if the characteristic triplet
  $(A,\nu,\gamma)$ of $(\xi,\zeta)$ satisfies the following conditions.
  \begin{itemize}
   \item[(1)] If $n\geq3$, the matrix $A$ satisfies
     \begin{equation}
       \label{eq:asi-quasi}
       a_{li}-a_{lj}=\thf(a_{ii}-a_{jj})
     \end{equation}
     for all $l\neq i,j$, $l\leq n$.
   \item[(2)] The image of $\hat\nu P^{-1}$ under $P$ of measure
     \begin{equation}
       \label{eq:nu-Essc-quasi}
       d\hat\nu(x)=e^{\langle v+\frac{1+\alpha}{2}e_i+\frac{1-\alpha}{2}e_j,x\rangle}d\nu(x)\,,
     \end{equation}
     is $\pi_{ij}$-invariant on $\HH\setminus\{0\}$.
  \item[(3)] $\gamma$ satisfies
    \begin{multline}
       \label{eq:drift-quasi-noncc}
        \gamma_i-\gamma_j=\thf(a_{jj}-a_{ii})
        -\frac{\alpha}{2}(a_{ii}+a_{jj}-2a_{ij})+\sum_{k=1}^{n+1}(a_{jk}-a_{ik})v_k\\
        +\int_{\R^{n+1}}(x_j-x_i)
        (e^{\langle v+\frac{1+\alpha}{2}e_i
        +\frac{1-\alpha}{2}e_j,x\rangle}\one_{\|Px\|\leq 1}-\one_{\|x\|\leq 1}) d\nu(x)\,.
\end{multline}
  \end{itemize}
\end{corollary}

For some applications, notably for semi-static hedging of barrier
options with unequal carrying costs, the symmetry should be
imposed on price changes adjusted with carrying costs. Unlike
equity markets, where the assumption of equal carrying costs is
often not totally unrealistic (e.g.\ in dividend-free cases), this
assumption is quite restrictive in currency markets, since the
risk-free interest rates in different countries usually differ.
The carrying costs on various assets amount to componentwise
multiplication of $\eta$ by a vector
$e^\lambda=(e^{\lambda_1},\dots,e^{\lambda_n})$, where
$\lambda_i=r-r_i$, $i=1,\dots,n$. In currency trading $r_i$
denotes the risk-free interest rate in the $i$th foreign market,
while in the share case it becomes the dividend yield of the $i$th
share. If useful, $\lambda$ can also have other interpretations
than being the pure carrying costs and $\eta$ need not be a
one-period martingale itself.

Multiplying $\eta$ with a vector representing unequal carrying
costs affects the (weighted) $ij$-swap-invariance
property~(\ref{eq:weighted-prop}). However, in some cases it is
possible to find $\alpha$ such that (\ref{eq:qsi}) holds, i.e.\
$\eta=e^{\xi+\lambda}$ is quasi-swap-invariant. In this case
$\xi+\lambda$ instead of $\xi$
satisfies~(\ref{eq:quasi-v-special-w}).
In the infinitely divisible case the only new condition on the
L\'evy triplet of $(\xi+\lambda,\zeta)$ concerns the ``drifts'':
\begin{multline}
  \label{eq:drift-quasi-cc}
  \gamma_i-\gamma_j=\thf(a_{jj}-a_{ii})
  -\frac{\alpha}{2}(a_{ii}+a_{jj}-2a_{ij})+\sum_{k=1}^{n+1}(a_{jk}-a_{ik})v_k\\
  +\int_{\R^{n+1}}(x_j-x_i)(e^{\langle v+\frac{1+\alpha}{2}e_i
        +\frac{1-\alpha}{2}e_j,x\rangle}\one_{\|Px\|\leq 1}-\one_{\|x\|\leq 1}) d\nu(x)+\lambda_j-\lambda_i
  \,.
\end{multline}
Note that this condition only depends on the carrying costs of the
$i$th and $j$th assets. 
If $\vartheta$ vanishes, then the condition on the drift
simplifies to
\begin{multline*}
  \gamma_i-\gamma_j=\thf(a_{jj}-a_{ii})
  -\frac{\alpha}{2}(a_{ii}+a_{jj}-2a_{ij})\\
  +\int_{\R^n}(x_j-x_i)(e^{\frac{1+\alpha}{2}x_i
    +\frac{1-\alpha}{2}x_j}\one_{\|P'x\|\leq 1}-\one_{\|x\|\leq 1}) d\nu(x)+\lambda_j-\lambda_i\,.
\end{multline*}

\begin{remark}[Determining $\alpha$ from the L\'evy triplet and the
  carrying costs]
  \label{rem:alpha}
  Consider $\eta=e^{\xi+\lambda}$ that satisfies (\ref{eq:qsi}) with
  given $\lambda$ and $\vartheta=\zeta$ such that $(\xi,\zeta)$ is
  infinitely divisible.  Note that neither (weighted)
  $ij$-swap-invariance nor the more general quasi-swap-invariance
  condition~(\ref{eq:qsi}) imply $\E e^{\xi_j}=1$. Thus, for many
  applications one additionally assumes that $\E
  e^{\xi_l}=\varphi_{(\xi,\zeta)}(-\imagi e_l)=1$ for all
  $l$ and also that $\E e^\zeta=1$. In particular, this implies
  \begin{equation}
    \label{eq:matr-for-alpha}
    \gamma_l=-\thf
    a_{ll}-\int_{\R^{n+1}}(e^{x_l}-1-x_l\one_{\|x\|\leq 1})d\nu(x)\,,
    \qquad l=i,j\,.
  \end{equation}
  Plug~(\ref{eq:matr-for-alpha}) in~(\ref{eq:drift-quasi-cc}) in order
  to see that $\alpha$ satisfies
  \begin{multline}
    \label{eq:levy-alpha}
    \alpha (a_{ii}+a_{jj}-2a_{ij})=2(\lambda_j-\lambda_i)+2(a_{j(n+1)}-a_{i(n+1)})\\
    +2\int_{\R^{n+1}}(e^{x_i}-e^{x_j}+(x_j-x_i)
    e^{\langle e_{n+1}+\frac{1+\alpha}{2}e_i
        +\frac{1-\alpha}{2}e_j,x\rangle}\one_{\|Px\|\leq 1})d\nu(x)\,.
  \end{multline}
  In the non-weighted case this condition can be written as
  \begin{multline*}
    \alpha(a_{ii}+a_{jj}-2a_{ij})\\=2(\lambda_j-\lambda_i)
    +2\int_{\R^n}(e^{x_i}-e^{x_j}+(x_j-x_i)
    e^{\frac{1+\alpha}{2}x_i+\frac{1-\alpha}{2}x_j}\one_{\|P'x\|\leq 1})d\nu(x)\,.
  \end{multline*}
  In the L\'evy processes setting the values of $\alpha$ calculated
  from the distributions at any time moment $t\geq0$ coincide.
\end{remark}

\begin{example}[Black--Scholes setting]
  \label{ex:q-two}
  In the absence of jumps it is easily possible to explicitly derive $\alpha$
  from~(\ref{eq:levy-alpha}).  Namely, if $\nu$ vanishes and $A$
  satisfies~(\ref{eq:levy-alpha}), with $a_{ii}+a_{jj}\neq 2a_{ij}$, then
  \begin{equation}
    \label{eq:rbs}
    \alpha=2\frac{\sum_{k=1}^{n+1}(a_{jk}-a_{ik})v_k+\lambda_j-\lambda_i}{a_{ii}+a_{jj}-2a_{ij}}\,,
  \end{equation}
  which in the non-weighted case simplifies to
  \begin{displaymath}
    \alpha=\frac{2(\lambda_j-\lambda_i)}{a_{ii}+a_{jj}-2a_{ij}}\,.
  \end{displaymath}
  In the bivariate Black--Scholes case this result has been derived
  in~\cite{schm10qf} by directly using a symmetry result in univariate Black--Scholes
  markets. Section~\ref{sec:hedg} shows that this
  result can be used for semi-statically hedging certain
  generalised swap-options in certain (in the bivariate case all)
  Black--Scholes economies.
\end{example}

Example~\ref{ex:q-two} demonstrates that turning to the more
general quasi-swap-invariance concept also in the equal carrying
cost case ($\lambda_i=\lambda_j$) yields considerably more
flexibility for modelling the asset prices. In particular for
e.g.\ $v=e_3$, each three-asset Black--Scholes model is
quasi-swap-invariant with $\alpha$ determined from (\ref{eq:rbs}).

\section{Hedging multi-asset barrier options}
\label{sec:hedg}

In this section we show how the analysed symmetry properties can
be used in order to create semi-static hedging strategies for
several multi-asset options. First we derive in
Section~\ref{sec:gen-hedge} a general hedging strategy for rather
general options in (weighted) quasi-swap-invariant models
extending results obtained in~\cite{schm10qf}, before applying them
to well-known options in Section~\ref{sec:ill-ex}. We will also
discuss examples where the more restrictive $ij$-exchangeability
property is needed.

It should be noted that the suggested hedging strategies are only
practicable provided that the instruments involved in the hedge
are liquid or can be replicated by liquid instruments.
Decompositions of not sufficiently liquid instruments in
over-the-counter traded claims is an active area of current
research and lies beyond the scope of this article. In this
relation, Carr and Laurence~\cite{car:laur09} write that ``all
major banks stand ready to provide over-the-counter quotes on
customised baskets'' and mention a decomposition possibility of
multivariate European payoff function in basket options, thereby
generalising results of Lipton~\cite{lip01} in the bivariate case,
see also \cite{baxt98,hen:sha90}. 
An easy decomposition formula for a large family of bivariate European
payoff functions in other over-the-counter traded instruments is given
in~\cite{sch:zue10}, while the results of \cite{bal:rut10} (in an
appropriately adjusted interpretation) yield further decompositions in
bivariate binary and certain bivariate correlation options.  There is
also a fast growing literature about sub- and super-replication of
basket options, see e.g.~\cite{lau:wan08,pen:say:ver:zul10} and the
literature cited therein. Furthermore, it is sometimes also possible
to increase the liquidity of the involved instruments by implementing
the hedges in a foreign derivative market or by using decomposition
methods, similarly to~\cite{eber:pap:shir08b} and \cite{schm10qf}.
Note that in special cases with equal (but not necessarily vanishing)
carrying costs our hedging instruments are already of the form of
exchange or basket options, see e.g.\ the swap-invariant version of
Example~\ref{ex:barrier-swap} or the non-quanto version of
Example~\ref{ex:hedg-on-ex}, respectively.

\subsection{A general hedging strategy}
  \label{sec:gen-hedge}

Consider a multivariate finite horizon model with the asset prices
given by
\begin{displaymath}
  (S_t,Z_t)=(S_0\circ e^{\tilde\xi_t},Z_0e^{\tilde\zeta_t})
  \,,\quad t\in[0,T]\,,
\end{displaymath}
where $(\tilde\xi_t,\tilde\zeta_t)=(\lambda t+\xi_t,\mu t+\zeta_t)$ is
a L\'evy process such that all components of $(e^{\xi_t},e^{\zeta_t})$
are martingales defined on a filtered probability space
$(\Omega,\salg,(\salg_t)_{t\in[0,T]},\Q)$ with the usual conditions
imposed on the filtration. Note that the vector
$(\lambda,\mu)\in\R^{n+1}$ represents deterministic carrying costs.

Fix $i<j$, $i,j\in\{1,\dots,n\}$, and assume that for every
$t\in[0,T]$, $(e^{\tilde\xi_t},e^{\tilde\zeta_t})$ satisfies the
quasi-swap-invariance property~(\ref{eq:qsi}) with
$\eta=e^{\tilde\xi_t}$ and $\vartheta=\tilde\zeta_t$, in particular,
this involves the integrability of $Z_tS_t$ and
$Z_tS_t(S_{ti}/S_{tj})^\alpha$ for all $t\in[0,T]$.

Note that in real market applications often neither the ``symmetry
assumption'' nor the assumption that the asset prices follow
multivariate componentwise exponentials of L\'evy processes will
typically be completely fulfilled. However, in the univariate case
several comparative studies, see
e.g.~\cite{car:wu06,eng:fen:nal:schw06,nal:pou06}, have confirmed
a relatively good performance of symmetry based semi-static
hedges, even if the assumptions behind the semi-static hedges are
not satisfied exactly.\footnote{We thank an anonymous referee for
this hint.}

Consider a payoff function $g\in\sH_1$ weighted by the terminal price
$Z_T$ of the $(n+1)$st asset, satisfying $\E |Z_Tg(S_T)|<\infty$, with
knock-in features given by the claims
\begin{equation}
  \label{eq:weighted-X-gen}
  X=Z_T g(S_T)\one_{\tau\leq T}\,,
\end{equation}
where
\begin{displaymath}
   \tau=\inf\{t : cS_{ti}\,\substack{\leq\\ \geq}\,S_{tj}\}\,.
\end{displaymath}
Note that we simultaneously handle the two knock-in cases
corresponding to the crossing of barrier from below or from above by
the ratio process $S_{tj}/S_{ti}$, choosing the appropriate inequality
in the indicator event. We assume that for the crossing from below
(resp. above) case the spot ratio $S_{0j}/S_{0i}$ lies below (resp.
above) the barrier.

Assume that the ratio process can not jump over the barrier $c$.
Then we can semi-statically replicate $X$ by the following (path
independent) European claim
\begin{equation}
  \label{eq:hedge-pf-gen}
  G(S_T,Z_T)=Z_T g(S_T)\one_B
  +Z_T g(\hat\kappa(c,S_T))\big(c\frac{S_{Ti}}{S_{Tj}}\big)^\alpha
  \one_{B_0}\,,
\end{equation}
where $B=\big\{cS_{Ti}\substack{\leq\\ \geq} S_{Tj}\big\}$,
$B_0=\big\{cS_{Ti}\substack{<\\>}S_{Tj}\big\}$ and
\begin{displaymath}
  \hat\kappa(c,S_T)=\big(S_{T1},\dots,S_{T(i-1)},
  \frac{S_{Tj}}{c},S_{T(i+1)},\dots,S_{T(j-1)},cS_{Ti},S_{T(j+1)},\dots,S_{Tn}\big)\,.
\end{displaymath}
In order to justify this hedge note that on the event $\{\tau>T\}$,
the claim in~(\ref{eq:hedge-pf-gen}) expires worthless as desired. If
the barrier knocks in, then at time $\tau$ we can
exchange~(\ref{eq:hedge-pf-gen}) for a claim on $Z_Tg(S_T)$ at zero
costs. To confirm this write
\begin{displaymath}
  Z_Tg(S_T)=Z_T g(S_T)\one_B+ Z_T g(S_T)\one_{B^c}\,,
\end{displaymath}
so we need to show that the conditional expectations of the second
term in the right-hand side given the stopping $\sigma$-algebra
$\salg_\tau$ coincides with the conditional expectation of the second
term in (\ref{eq:hedge-pf-gen}) on the event $\{\tau\leq T\}$.

Since $(\xi_t,\zeta_t)$ is a L\'evy process,
$((\xi_\tau,\zeta_\tau),(\xi_T,\zeta_T))$ and
\begin{displaymath}
  ((\xi_\tau,\zeta_\tau),(\xi_\tau+\xi'_{T-\tau},\zeta_\tau+\zeta'_{(T-\tau)}))
\end{displaymath}
share the same distribution on the event $\{\tau\leq T\}$, where
$(\xi'_t,\zeta'_t)$, $t\in[0,T]$, is an independent copy of the
process $(\xi_t,\zeta_t)$, $t\in[0,T]$. Hence,
$((S_\tau,Z_\tau),(S_T,Z_T))$ and $((S_\tau,Z_\tau),(S_\tau\circ
\eta'_\sigma,Z_\tau Z'_\sigma)$ also coincide in distribution, where
$(\eta'_t,Z'_t)$, is an independent copy of the process $(\eta_t,Z_t)$
and $\sigma=T-\tau$. The quasi-swap-invariance
property~(\ref{eq:pos-quasi}) together with
Remark~\ref{rem:tau-T-mult} yield that
\begin{align*}
  \E[Z_Tg(S_T)\one_{B^c}|\salg_\tau]
  &=\E\big[Z_\tau Z'_\sigma g(S_\tau\circ \eta'_\sigma)
  \one_{\{cS_{\tau i}\eta'_{\sigma i}\,\substack{>\\<}\,
    S_{\tau j}\eta'_{\sigma j}\}}|\salg_\tau\big]\\
  &=\E\Big[Z_\tau Z'_\sigma g(S_\tau\circ\pi_{ij}\eta'_\sigma)
  \one_{\{cS_{\tau i}\eta'_{\sigma j}\,\substack{>\\<}\,
    S_{\tau j}\eta'_{\sigma i}\}}
  \Big(\frac{\eta'_{\sigma i}}{\eta'_{\sigma j}}\Big)^\alpha|\salg_\tau\Big]\\
  &=\E\Big[Z_\tau Z'_\sigma g(\hat\kappa(c,S_\tau\circ \eta'_\sigma))
  \one_{\{S_{\tau j}\eta'_{\sigma j}\,\substack{>\\<}\,
    cS_{\tau i}\eta'_{\sigma i}\}}
  \Big(\frac{cS_{\tau i}\eta'_{\sigma i}}{S_{\tau_j}
    \eta'_{\sigma j}}\Big)^\alpha|\salg_\tau\Big]\\
  &=\E\Big[Z_T g(\hat\kappa(c,S_T))\big(c\frac{S_{Ti}}{S_{Tj}}\big)^\alpha
  \one_{B_0}|\salg_\tau\Big]
\end{align*}
on the event $\{\tau\leq T\}$. Note that we have used that $S_{\tau
  j}=cS_{\tau i}$.  Hence, on the event $\{\tau\leq T\}$
\begin{displaymath}
  \E[Z_T g(S_T)|\salg_\tau]=\E[G(S_T,Z_T)|\salg_\tau]\,.
\end{displaymath}
The above arguments also can be used to valuate the described barrier
options for models with continuous sample paths in the $i$th and $j$th
component.

Note that if the ratio-process can jump over the barrier, the hedge
in~(\ref{eq:hedge-pf-gen}) is no longer exact.

\subsection{Illustrative examples}
\label{sec:ill-ex}

We will assume in all examples without loss of generality that $i=1$
and $j=2$, and so accordingly speak about $12$-swap-invariance or
$12$-exchangeability.

\begin{example}[Barrier quanto-swap options]
  \label{ex:barrier-quanto-swap}
  Consider a vector of asset prices
  \begin{displaymath}
    S_t=(S_{t1},S_{t2},S_{t3})
    =(S_{01}e^{\lambda_1 t}e^{\xi_{t1}},S_{02}e^{\lambda_2
    t}e^{\xi_{t2}}
    ,S_{03}e^{\lambda_3 t}e^{\xi_{t3}})=S_0\circ e^{\xi_t+\lambda t}
  \end{displaymath}
  with $\lambda=(\lambda_1,\lambda_2,\lambda_3)$ representing the
  carrying costs and
  \begin{displaymath}
    \eta_t=(\eta_{t1},\eta_{t2},\eta_{t3})
     =(e^{\lambda_1t+\xi_{t1}},e^{\lambda_2t+\xi_{t2}},e^{\lambda_3t+\xi_{t3}})=e^{\xi_t+\lambda t}\,,
  \end{displaymath}
  being $12$-quasi-swap-invariant (with the weight given by the third
  asset) for all $t\in [0,T]$, where other conditions remain the same
  as in Section~\ref{sec:gen-hedge}. Consider barrier claims defined
  by
  \begin{align*}
    X_{\mathrm{qsw}} &=S_{T3}(aS_{T1}-bS_{T2})_+\,\one_{\exists
      t\in[0,T],\,
      cS_{t1}\leq S_{t2}}\,,\\
    Y_{\mathrm{qsw}} &=S_{T3}(aS_{T1}-bS_{T2})_+\one_{cS_{t1}>S_{t2}\forall t\in[0,T]}\,,
  \end{align*}
  where $cS_{01}>S_{02}$, the positive parameters satisfy $0<a\leq
  bc$, and
  \begin{equation}
    \label{eq:tau-def}
    \tau=\inf\{t : cS_{t1}\leq S_{t2}\}\,.
  \end{equation}
  By~(\ref{eq:hedge-pf-gen}) the hedge portfolio for
  $X_{\mathrm{qsw}}$ is given by
  \begin{multline*}
    G(S_{T1},S_{T2},S_{T3})\\=S_{T3}(aS_{T1}-bS_{T2})_+\one_{cS_{T1}\leq S_{T2}}
      +S_{T3}(\frac{a}{c}S_{T2}-bcS_{T1})_+
      \Big(c\frac{S_{T1}}{S_{T2}}\Big)^{\alpha}\one_{cS_{T1}< S_{T2}}\,.
  \end{multline*}
  Since $0<a\leq bc$, if $cS_{T1}\leq S_{T2}$, then
  $(aS_{T1}-bS_{T2})_+$ is out of the money so that the first term
  vanishes. Furthermore, $0<a\leq bc$ implies that if
  $(\frac{a}{c}S_{T2}-bcS_{T1})_+$ is (strictly) positive, then
  $cS_{T1}<S_{T2}$, so that we can omit the indicator function in the
  second term. Thus, $X_{\mathrm{qsw}}$ can be hedged (exactly if the
  ratio process can not jump over the barrier) with a long position in
  the European derivative with payoff
  \begin{displaymath}
    S_{T3}(\frac{a}{c}S_{T2}-bcS_{T1})_+\big(c\frac{S_{T1}}{S_{T2}}\big)^{\alpha}\,.
  \end{displaymath}
  By the knock-in knock-out parity $Y_{\mathrm{qsw}}$ can be hedged by
  a short position in this derivative and a long position in the
  European derivative with payoff $S_{T3}(aS_{T1}-bS_{T2})_+$. In the
  weighted $12$-swap-invariant case $\alpha=0$ and the hedging
  instruments reduce to weighted quanto-swap options.
\end{example}

\begin{example}[Barrier swap options]
  \label{ex:barrier-swap}
  By specialising the claims $X_{\mathrm{qsw}}$ and
  $Y_{\mathrm{qsw}}$ to the cases where $S_{T3}=1$ (with other assumptions
  unchanged) we get weighted barrier swap (also known as Margrabe)
  options with knocking conditions, defined by the following claims
  \begin{align*}
    X_{\mathrm{sw}} &=(aS_{T1}-bS_{T2})_+\,\one_{\exists t\in[0,T],\, cS_{t1}\leq S_{t2}}\,,\\
    Y_{\mathrm{sw}} &=(aS_{T1}-bS_{T2})_+\one_{cS_{t1}>S_{t2}\forall
    t\in[0,T]}\,,
  \end{align*}
  where the assumptions on the parameters $a$, $b$, $c$ remain
  unchanged.  Assuming $(\eta_{t1},\eta_{t2})$ to be
  quasi-swap-invariant, the hedging portfolio for $X_{\mathrm{sw}}$
  consists in a long position in the European derivative with payoff
  $(\frac{a}{c}S_{T2}-bcS_{T1})_+(c\frac{S_{T1}}{S_{T2}})^{\alpha}$,
  while the hedge of $Y_{\mathrm{sw}}$ is given by a short position in
  this derivative and a long position in the European derivative with
  payoff $(aS_{T1}-bS_{T2})_+$. In the swap-invariant case all hedging
  instruments reduce to weighted swap-options, respectively weighted
  Margrabe options.
\end{example}

\begin{example}[Hedges based on the exchangeability property]
  \label{ex:hedg-on-ex}
  We end this section by discussing an example where we need more
  symmetry in the model than (weighted) swap-invariance in order to
  hedge some basket payoffs with barrier features on a ratio-process.
  Assume that the vector of asset prices
  from Example~\ref{ex:barrier-quanto-swap}
  $S_t=(S_{t1},S_{t2},S_{t3})$ is such that $(e^{\xi_{t1}},e^{\xi_{t2}},e^{\xi_{t3}})$
  is $12$-exchangeable for all $t\in
  [0,T]$, while the remaining assumptions are not changed.  Let the
  carrying costs $\lambda_1=\lambda_2$ be the same for the first and
  the second assets, e.g.\ both being the risk-free interest rate.
  Assume $cS_{01}>S_{02}$ and define the stopping time $\tau$ by
  (\ref{eq:tau-def}). Consider the claim
  \begin{displaymath}
     Y_{\mathrm{qsp}}=S_{T3}\big(aS_{T1}-bS_{T2}-k\big)_+
      \one_{cS_{t1}>S_{t2}\,\forall t\in [0,T]}\,,
  \end{displaymath}
  with positive weights $a,b,c$, $a\leq bc$ and non-negative strike
  $k$.  This option is knocked out if the ratio
  $\frac{S_{t2}}{S_{t1}}$ achieves or exceeds $c$.

  We again assume for a moment that jumps cannot cross the barrier,
  e.g.\ being the case in the Black--Scholes setting. By similar
  arguments as in Example~\ref{sec:gen-hedge}, we can hedge the claim
  $Y_{\mathrm{qsp}}$ by taking the following positions in the European
  weighted quanto-spread options
  \begin{align*}
    \text{long}\qquad &   S_{T3}\big(aS_{T1}-bS_{T2}-k\big)_+\,,\\
    \text{short}\qquad & S_{T3}\big(\frac{a}{c}S_{T2}-bcS_{T1}-k\big)_+\,.
  \end{align*}
  This also yields that the knock-in claim
  \begin{displaymath}
    X_{\mathrm{qsp}}=S_{T3}\big(aS_{T1}-bS_{T2}-k\big)_+
    \one_{\exists t\in[0,T],\, cS_{t1}\leq S_{t2}}\,,
  \end{displaymath}
  with the same parameters $a,b,c,k$, can be hedged with a long position
  in the European option given by the payoff function
  \begin{displaymath}
    S_{T3}\big(\frac{a}{c}S_{T2}-bcS_{T1}-k\big)_+\,.
  \end{displaymath}
  In case of jump processes the exchangeability implies that $e^{\xi_{t1}}$
  has non-problematic up (problematic down) jumps if and only if
  $e^{\xi_{t2}}$ has problematic up (non-problematic down) jumps, so that
  $cS_{\tau 1}$ is no longer almost surely equal $S_{\tau 2}$. This
  fact leads to a super-replication of knock-in options and a more
  problematic sub-replication of knock-out options.
\end{example}

\section*{Acknowledgements}

The authors are grateful to anonymous referees for careful reading
this paper and suggesting improvements.

\newcommand{\noopsort}[1]{} \newcommand{\printfirst}[2]{#1}
  \newcommand{\singleletter}[1]{#1} \newcommand{\switchargs}[2]{#2#1}

\end{document}